\documentclass[runningheads]{llncs}
\usepackage[T1]{fontenc}
\usepackage{amsmath,amssymb}
\usepackage{booktabs,array}
\usepackage{graphicx}
\usepackage{tikz}
\usetikzlibrary{arrows.meta,positioning,calc}
\usepackage[hidelinks]{hyperref}
\usepackage{microtype}

\newcommand{\I}{\mathcal I}
\newcommand{\R}{\mathbb R}
\newcommand{\Z}{\mathbb Z}

\newcommand{\up}{\operatorname{up}_2}
\definecolor{spine}{RGB}{158,42,43}
\definecolor{component}{RGB}{46,88,128}

\begin{document}
\title{Almost Linear Universal Point Sets\protect\\ for Planar Graphs}
\titlerunning{Almost Linear Universal Point Sets for Planar Graphs}
\author{Taylor Gordon}
\authorrunning{T. Gordon}
\institute{}

\maketitle

\begin{abstract}
A point set is universal for planar graphs on $n$ vertices if every such graph
has a straight-line drawing without crossings whose vertices belong to the set.
We construct universal point sets of size $n^{1+o(1)}$, improving
the previous quadratic upper bound. Our construction uses the reduction of
Bannister, Cheng, Devanny, and Eppstein from universal point sets to
superpatterns for $213$-avoiding permutations. We represent these permutations
by ordered rooted forests and construct a small family of intervals containing
every such forest. The result follows from a straightforward bound on the size
of the family of intervals.
GPT-6~Astra assisted in developing the construction and proof.
\end{abstract}

\section{Introduction}
\label{sec:introduction}
A \emph{universal point set} for planar graphs on $n$ vertices is a finite set
$P\subset\R^2$ such that every planar graph on $n$ vertices has a straight-line
drawing without crossings in which its vertices map to distinct points of $P$.
The mapping may depend on the graph. Let $u(n)$ be the smallest possible size
of such a set.

How small can a universal point set be? The question is attributed to Mohar
and appears as Problem~45 in The Open Problems Project~\cite{topp}.
In particular, it asks whether $u(n)=O(n)$. F\'ary's theorem~\cite{fary} gives a
straight-line drawing of each individual planar graph, but does not give a
small point set that works for all of them. The grid-drawing results of
de~Fraysseix, Pach, and Pollack~\cite{dfpp} and Schnyder~\cite{schnyder},
published in 1990, established $u(n)=O(n^2)$. Their constructions place every planar graph on a
grid with a linear number of rows and columns.

Subsequent work reduced the leading constant in the quadratic bound.
A triangular subset of the de~Fraysseix--Pach--Pollack grid gives
$n^2-O(n)$ points, and Schnyder's construction gives
$n^2/2-O(n)$ points. In 2008, Brandenburg~\cite{brandenburg} obtained
$4n^2/9+O(n)$ points. Bannister, Cheng, Devanny, and
Eppstein~\cite{bcde} obtained $n^2/4-\Theta(n)$ points in 2014 by connecting
the problem to permutation patterns. The quadratic order remained the best
general upper bound in the 2025 survey of Dujmovi\'c and Morin~\cite{dm}.

Smaller sets are known for several restricted classes. An arbitrary set of
$n$ points in general position suffices for outerplanar graphs~\cite{gritzmann}.
Bannister et~al.\ obtained $O(n\log n)$ points for simply nested planar graphs
and $n\log^{O(1)}n$ points for each fixed bound on pathwidth~\cite{bcde}.
Angelini et~al.~\cite{angelini} obtained $O(n\log n)$ points for
$2$-outerplanar graphs. For planar $3$-trees, Fulek and T\'oth~\cite{fulek}
gave a bound of $O(n^{3/2}\log n)$.
Felsner et~al.~\cite{felsner} proved that
$2n-2$ points suffice for bipartite planar graphs and for planar graphs of
maximum degree three. Table~\ref{tab:bounds} collects representative bounds;
all of them concern drawings with straight edges.

\begin{table}[t]
\caption{Selected upper bounds on the size of universal point sets for
$n$-vertex graphs. Constants in the bounded-pathwidth row may depend on the
fixed pathwidth bound. The final row is the result of this paper.}
\label{tab:bounds}
\centering
\small
\setlength{\tabcolsep}{3.4pt}
\renewcommand{\arraystretch}{1.17}
\begin{tabular}{@{}>{\raggedright\arraybackslash}p{.315\textwidth}>{\raggedright\arraybackslash}p{.315\textwidth}>{\raggedright\arraybackslash}p{.315\textwidth}@{}}
\toprule
Graph class & Number of points & Reference \\
\midrule
All planar graphs & $n^2-O(n)$ & de Fraysseix et al.~\cite{dfpp} \\
All planar graphs & $n^2/2-O(n)$ & Schnyder~\cite{schnyder} \\
All planar graphs & $4n^2/9+O(n)$ & Brandenburg~\cite{brandenburg} \\
All planar graphs & $n^2/4-\Theta(n)$ & Bannister et al.~\cite{bcde} \\
Outerplanar graphs & $n$ & Gritzmann et al.~\cite{gritzmann} \\
Simply nested graphs & $O(n\log n)$ & Bannister et al.~\cite{bcde} \\
Planar graphs of bounded pathwidth & $n\log^{O(1)}n$ & Bannister et al.~\cite{bcde} \\
$2$-outerplanar graphs & $O(n\log n)$ & Angelini et al.~\cite{angelini} \\
Planar $3$-trees & $O(n^{3/2}\log n)$ & Fulek and T\'oth~\cite{fulek} \\
Bipartite planar graphs;
planar graphs of maximum degree $3$ & $2n-2$ & Felsner et al.~\cite{felsner} \\
\midrule
\textbf{All planar graphs} & $\boldsymbol{n^{1+o(1)}}$ & \textbf{This paper} \\
\bottomrule
\end{tabular}
\end{table}

The lower bounds are linear. Chrobak and Karloff~\cite{chrobak} obtained a
lower bound exceeding $n$ by a constant factor. Kurowski~\cite{kurowski}
proved that at least $1.235n$ points are necessary for sufficiently large
$n$, and Scheucher, Schrezenmaier, and
Steiner~\cite{scheucher} proved the bound $(1.293-o(1))n$.
Thus a substantial gap remained between the general lower and upper bounds.
Our main result reduces this gap to a subpolynomial factor.

\begin{theorem}
\label{thm:main}
For every positive integer $n$, there is an explicitly constructible universal
point set for planar graphs on $n$ vertices of size
\[
 n\exp\bigl(O(\sqrt{\log n})\bigr)=n^{1+o(1)}.
\]
\end{theorem}

This establishes $u(n)=n^{1+o(1)}$. It does not settle whether $u(n)=O(n)$.
The points need not form a grid, and the bound concerns their number, rather
than the area of the drawing or the magnitude of the coordinates.

Our proof begins with the reduction of Bannister et~al.~\cite{bcde}.
A permutation containing every $213$-avoiding permutation of length $n-3$
can be converted into a universal point set by adding three points and
stretching one coordinate. We therefore construct a small permutation with
this containment property.
The $213$-avoiding permutations correspond to ordered rooted forests, which
can be represented by nested intervals. Our interval family has $n$ layers.
Layer $i$ has one prescribed radius, and its midpoints are spaced by the
smallest power of two at least $i$. A vertex whose subtree has $i$ vertices
uses an interval from layer $i$. A child containing almost all of its parent's
subtree stays concentric with the parent; the other children fit beside it.
When no child is this large, all children fit by ordinary ordered packing.
The size bound then follows by summing the number of midpoints over the
dyadic layers.

Section~\ref{sec:forests} describes the two reductions.
Section~\ref{sec:packing} defines the layers and gives the packing estimates.
Section~\ref{sec:embedding} represents trees and forests using these layers.
Section~\ref{sec:count} counts the interval family and
completes the proof. For completeness, Appendix~\ref{app:reduction} gives
the point-set construction and the planarity argument from~\cite{bcde}.

A \href{https://github.com/taylorgordon20/math/tree/main/213-superpatterns}{Lean~4 formalization}
of an earlier interval construction proves the same
$n^{1+o(1)}$ superpattern bound.
The layered proof presented here is not part of that formalization.

This paper was prepared with substantial assistance from GPT-6 Astra.
Starting from the author's interval-based approach, the model helped develop
the construction, formulate and revise the proof, check finite examples, and
draft the exposition.

\section{Superpatterns and Ordered Forests}
\label{sec:forests}
A permutation $\pi$ is a \emph{pattern} of a permutation $\sigma$ if a
subsequence of $\sigma$ has the same relative order as $\pi$.
A permutation \emph{avoids $213$} if it has no indices $i<j<k$ with
$\pi_j<\pi_i<\pi_k$.
Write $S_n(213)$ for the set of such permutations of length $n$.
An \emph{$S_n(213)$-superpattern} contains every member of $S_n(213)$.
The superpattern itself is not required to avoid $213$.

The following lemma is the connection to universal point sets.

\begin{lemma}[Bannister, Cheng, Devanny, and Eppstein~\cite{bcde}]
\label{lem:bcde}
Let $n\ge3$. An $S_{n-3}(213)$-superpattern of length $M$ gives a
universal point set of size $M+3$ for planar graphs on $n$ vertices.
\end{lemma}

The construction in Lemma~\ref{lem:bcde} adds three entries to the
permutation and maps an entry of rank $j$ in position $i$ to
$(i,(M+3)^j)$. The choice of which points represent the graph comes from
two traversals of a canonical spanning tree. The exponential stretching
ensures that these choices give a drawing without crossings. We use the
lemma as a reduction; its details are given in Appendix~\ref{app:reduction}.

An \emph{ordered rooted forest} is a rooted forest with an order on its roots
and on the children of each vertex. An \emph{interval representation} of
such a forest assigns an interval to each vertex so that a descendant's
interval lies in the interior of its ancestor's interval, and intervals of
incomparable vertices are disjoint and occur in the prescribed order.
We require a positive gap between disjoint intervals.

\begin{lemma}
\label{lem:forest}
Suppose a family of $M$ intervals contains an interval representation of every
ordered rooted forest on $n$ vertices. Then it gives an
$S_n(213)$-superpattern of length $M$.
\end{lemma}
\begin{proof}
Let $\pi\in S_n(213)$. Order its entries by
\[
 i\prec j\quad\Longleftrightarrow\quad i<j\ \text{and}\ \pi_i<\pi_j.
\]
The predecessors of each entry form a chain. Otherwise two incomparable
predecessors, followed by that entry, would form a $213$ pattern.
Every nonminimal entry therefore has a unique immediate predecessor, so the
cover relations form a rooted forest.

Order its roots and children by their indices in $\pi$. The subtrees do not
interleave in this order. To see this, suppose $u<v$ are incomparable and a
descendant $w$ of $u$ occurs after $v$. Then
$\pi_v<\pi_u<\pi_w$, so $v\prec w$. This contradicts the fact that the
predecessors of $w$ form a chain. Thus the index order is the preorder of an
ordered rooted forest.

Take a representation of this forest in the given family. If the interval
of $v$ is $[\ell_v,r_v]$, then the points $(\ell_v,-r_v)$ have pattern
$\pi$: ancestor pairs increase in both coordinates, while incomparable
subtrees have opposite orders in the two coordinates.

Now order the entire interval family by increasing left endpoint, and give
each interval the rank of its right endpoint in decreasing order. Break
ties by fixed rules. These ranks form a permutation of length $M$.
All endpoints in a selected forest representation are distinct, so the
tie-breaking does not change its pattern. Hence this permutation contains
every member of $S_n(213)$.
\qed
\end{proof}

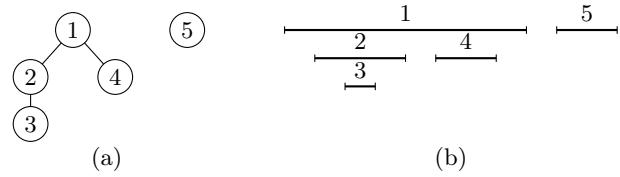
\begin{figure}[t]
\centering
\begin{tikzpicture}[x=0.56cm,y=0.62cm,every node/.style={font=\small}]
\begin{scope}[xshift=-2.8cm]
\draw (0,2)--(-1,1)--(-1,0);
\draw (0,2)--(1,1);
\foreach \x/\y/\z in {0/2/1,-1/1/2,-1/0/3,1/1/4,2.7/2/5}
  \node[circle,draw,fill=white,inner sep=2pt] at (\x,\y) {$\z$};
\node at (.8,-.75) {(a)};
\end{scope}
\begin{scope}[xshift=1.6cm,x=0.4cm]
\foreach \a/\b/\y/\z in {-4/4/2/1,-3/0/1.4/2,-2/-1/.8/3,1/3/1.4/4,5/7/2/5}{
 \draw[thick] (\a,\y)--(\b,\y);
 \draw (\a,\y-.08)--(\a,\y+.08);
 \draw (\b,\y-.08)--(\b,\y+.08);
 \node[above] at ({(\a+\b)/2},\y) {$\z$};
}
\node at (1.5,-.75) {(b)};
\end{scope}
\end{tikzpicture}
\caption{An ordered rooted forest (a) and an interval representation (b).
The vertices are numbered in preorder. Reading the ranks of decreasing right
endpoints in this order gives the permutation $24531$. The vertical positions
of the intervals are used only to separate them in the figure.}
\label{fig:forest}
\end{figure}

\section{Layered Intervals}
\label{sec:packing}
For $x\ge1$, put
\[
 s(x)=\sqrt{2\log_2x+4},\qquad A(x)=2^{s(x)},\qquad F(x)=2xA(x).
\]
For a positive integer $i$, define
\begin{equation}
 R_i=\lfloor F(i)\rfloor,
 \qquad c_i=\up(i)=2^{\lceil\log_2i\rceil}.
 \label{eq:layers}
\end{equation}
Thus $R_1=8$, $i\le c_i<2i$, and $c_i\mid c_j$ whenever $i\le j$.
An interval from \emph{layer $i$} has radius $R_i$ and midpoint in $c_i\Z$.
We will use layer $i$ at a vertex whose subtree has $i$ vertices.
Consequently, translating a tree representation by a multiple of its root's
spacing preserves the grids of every descendant.

For $n\ge1$, set $W=\up(R_n)$ and define the finite family
\begin{equation}
 \I_n=\bigl\{[jc_i-R_i,jc_i+R_i]:
 1\le i\le n,\quad 0\le j\le 2W/c_i,\quad j\in\Z\bigr\}.
 \label{eq:family}
\end{equation}
The upper limit is an integer: $c_i$ and $W$ are powers of two and
$c_i\le W$. We first construct representations near zero, then translate
them by $W$ to obtain intervals in~\eqref{eq:family}.

\begin{lemma}[Ordered packing]
\label{lem:packing}
Suppose $T_1,\ldots,T_q$ have interval representations with root midpoint
zero, each vertex using the layer indexed by its subtree size.
Let the tree sizes be $a_1,\ldots,a_q$, and put
\[
 B=\sum_{j=1}^q(2R_{a_j}+c_{a_j}).
\]
For any integer $z$, these representations can be translated into
$[z+1,z+B]$ in the prescribed order, with positive gaps, while preserving
all their layers and grids. They can also be packed into $[z-B,z-1]$.
\end{lemma}
\begin{proof}
Start with $x_0=z+1$ and choose the successive root midpoints
\[
 u_j=c_{a_j}\left\lceil\frac{x_{j-1}+R_{a_j}}{c_{a_j}}\right\rceil,
 \qquad x_j=u_j+R_{a_j}+1.
\]
The left endpoint is at least $x_{j-1}$. Integer rounding gives
$x_j\le x_{j-1}+2R_{a_j}+c_{a_j}$, so the last right endpoint is at
most $z+B$. Each translation preserves every descendant's grid.
For packing to the left, process the trees in reverse order, starting at
$x=z-1$, and use $u=c_a\lfloor(x-R_a)/c_a\rfloor$, followed by
$x=u-R_a-1$. The individual representations are translated, not reflected.
For an empty list take $B=0$.
\qed
\end{proof}

We record the estimates that explain the choice of radii. Extend $F$ to
$[0,1]$ by $F(x)=8x$. It is increasing and convex on $[0,\infty)$, with
$F(0)=0$. Indeed, for $x\ge1$,
\begin{equation}
 F'(x)=2A(x)\left(1+\frac1{s(x)}\right),\qquad
 F''(x)=\frac{2A(x)}{x s(x)^3}
 \left(s(x)^2+s(x)-\frac1{\ln2}\right)>0.
 \label{eq:derivatives}
\end{equation}
The right derivative at $1$ is $12$, exceeding the left derivative $8$.
In particular, $F'(x)\ge12$ for $x\ge1$, and convexity implies
$\sum_jF(a_j)\le F(\sum_j a_j)$ for nonnegative $a_j$.

\begin{lemma}[Space between layers]
\label{lem:slack}
Let $t\ge10$ be an integer, and write $s=s(t)$, $A=A(t)$, and
$\delta=t/A$. Then $1\le\delta\le t/4$ and
\begin{align}
 A(a)&\le A/2 &&(1\le a\le\delta), \label{eq:small}\\
 F'(x)&\ge2A+2 &&(t-\delta\le x\le t). \label{eq:large}
\end{align}
Moreover, if $a_j\ge0$, $\sum_j a_j\le t$, and $a_j\le t-\delta$,
then
\begin{equation}
 \sum_jF(a_j)\le F(t-\delta)+F(\delta)
 \le F(t)-t-2\delta.
 \label{eq:balanced}
\end{equation}
\end{lemma}
\begin{proof}
The inequality $t\ge A(t)$ is equivalent to
$\log_2t\ge1+\sqrt5$, and holds for $t\ge10$. Also $A\ge4$.
If $a\le\delta$, then
$s(a)^2\le s^2-2s<(s-1)^2$, proving~\eqref{eq:small}.
For $t-\delta\le x\le t$, we have
\[
 \frac{A(x)}A
 =\left(\frac xt\right)^{2/(s(x)+s)}\ge\frac xt,
\]
since $x\ge1$ and $s(x),s\ge2$. Thus $A(x)\ge A-1$, and
\[
 F'(x)\ge2(A-1)(1+1/s)\ge2A+2.
\]
The last inequality uses $2^s\ge2s+1$ for $s\ge3$.
Convexity maximizes $\sum_jF(a_j)$, subject to its stated constraints,
by concentrating the mass into $t-\delta$ and $\delta$; here
$t-\delta\ge t/2$. Finally,
\[
 F(t)-F(t-\delta)-F(\delta)
 \ge(2A+2)\delta-A\delta=t+2\delta,
\]
by~\eqref{eq:small} and~\eqref{eq:large}.
\qed
\end{proof}

\section{Embedding Trees and Forests}
\label{sec:embedding}
\begin{lemma}
\label{lem:tree}
Every ordered rooted tree on $t$ vertices has a layered interval
representation whose root is $[-R_t,R_t]$. Each vertex uses the layer
indexed by its subtree size.
\end{lemma}
\begin{proof}
Induct on $t$. A single vertex uses $[-8,8]$. Recursively represent the
children, whose sizes $a_j$ sum to $t-1$, and write
$B=\sum_j(2R_{a_j}+c_{a_j})$ for their packing capacity.

First suppose $2\le t\le9$. By convexity and~\eqref{eq:derivatives},
\[
 B\le2F(t-1)+2(t-1)\le2F(t)-8<2R_t-1.
\]
Packing after $-R_t$ therefore puts the children strictly inside the
prescribed root interval.

Now let $t\ge10$, with $A$ and $\delta$ as in Lemma~\ref{lem:slack}.
Call a child \emph{dominant} if its size $m$ exceeds $t-\delta$.
There is at most one. If there is no dominant child,
\eqref{eq:balanced} gives
\[
 B\le2\sum_jF(a_j)+2(t-1)
 \le2F(t)-4\delta-2<2R_t-1.
\]
Again pack after $-R_t$.

If a dominant child exists, keep its root midpoint at zero. The other
children have total size $k=t-m-1<\delta$. Their combined capacity
$B_* = B^-+B^+$, for the children before and after the dominant child,
satisfies
\[
 B_*\le\sum_{a_j\ne m}(2F(a_j)+2a_j)\le(2A+2)k.
\]
On the other hand,~\eqref{eq:large} gives
\[
 F(t)-F(m)\ge(2A+2)(k+1)>B_*+2.
\]
Hence $R_t-R_m\ge B_*+1$. Pack the preceding children to the left of
$-R_m$ and the following children to the right of $R_m$. They lie
strictly inside $[-R_t,R_t]$, in the required order and with positive
gaps. The root has midpoint zero, so it belongs to its prescribed grid.
This completes the induction.
\qed
\end{proof}

\begin{figure}[t]
\centering
\begin{tikzpicture}[x=.7cm,y=.7cm,every node/.style={font=\small}]
\begin{scope}[xshift=-2.7cm]
\draw[spine,very thick] (0,2.6)--(0,0.6);
\foreach \y/\v in {2.6/v_1,1.6/v_2,.6/v_s}
 \node[circle,draw=spine,fill=white,inner sep=2.2pt] at (0,\y) {$\v$};
\draw (0,2.6)--(-1.1,2.05) (0,1.6)--(1.1,1.05);
\draw[component,fill=component!8] (-1.1,2.05)--(-1.55,1.45)--(-.65,1.45)--cycle;
\draw[component,fill=component!8] (1.1,1.05)--(.65,.45)--(1.55,.45)--cycle;
\draw (0,.6)--(-.45,-.05) (0,.6)--(.45,-.05);
\node at (0,-.65) {(a)};
\end{scope}
\begin{scope}[xshift=1.7cm]
\draw[densely dashed,gray] (0,-.05)--(0,3.05);
\foreach \r/\y/\v in {2.7/2.6/v_1,1.9/1.6/v_2,1.1/.6/v_s}{
 \draw[spine,thick] (-\r,\y)--(\r,\y);
 \draw[spine] (-\r,\y-.1)--(-\r,\y+.1) (\r,\y-.1)--(\r,\y+.1);
 \node[above] at (0,\y) {$\v$};
}
\draw[component,fill=component!8] (-2.55,1.96) rectangle (-2.05,2.22);
\draw[component,fill=component!8] (1.25,.96) rectangle (1.75,1.22);
\draw[component,thick] (-.95,.13)--(-.18,.13) (.13,.13)--(.93,.13);
\node at (0,-.65) {(b)};
\end{scope}
\end{tikzpicture}
\caption{A path obtained by repeatedly following a dominant child (a),
and its interval placement (b). The path intervals are concentric. Other
children lie between successive path intervals, or inside the final one. The vertical positions do not form part of the representation.}
\label{fig:spine}
\end{figure}
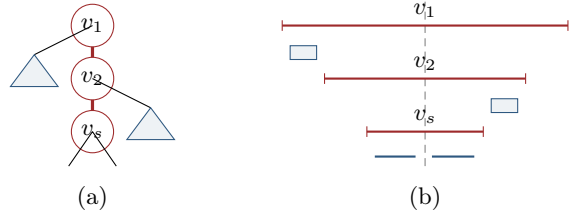

\begin{lemma}
\label{lem:forest-layers}
Every ordered rooted forest on $n\ge1$ vertices has a layered interval
representation contained in $[-R_n,R_n]$, using only layers $1,\ldots,n$.
\end{lemma}
\begin{proof}
Represent its trees by Lemma~\ref{lem:tree}. A forest with a single tree
already has the required representation. Let $a_j$ be the root subtree
sizes and $B=\sum_j(2R_{a_j}+c_{a_j})$.

Suppose first that $n\ge10$, and put $A=A(n)$ and $\delta=n/A$.
If every $a_j\le n-\delta$, then~\eqref{eq:balanced} yields
$B\le2F(n)-4\delta<2R_n+1$. Pack after $-R_n-1$; all intervals lie
in $[-R_n,R_n]$.
Otherwise, center the unique dominant root, of size $m>n-\delta$, at
zero. The other roots have total size $k=n-m<\delta$ and combined
capacity $B_*\le(2A+2)k$. Now~\eqref{eq:large} gives
$F(n)-F(m)\ge(2A+2)k\ge B_*$. Since $B_*$ is an integer,
$R_n-R_m\ge B_*$. Pack the roots before and after the dominant one to
its left and right. They fit in the closed envelope $[-R_n,R_n]$.
No additional root interval is needed.

It remains to consider $2\le n\le9$ and at least two roots.
Convexity gives $\sum_jF(a_j)\le F(n-1)+F(1)=F(n-1)+8$.
For $n\le4$, the difference $F(n)-F(n-1)-8$ is at least $4\ge n$.
For $5\le n\le9$, it exceeds $F'(4)-8>10\ge n$.
Thus in either case
\[
 B\le2\sum_jF(a_j)+2n\le2F(n)<2R_n+2.
\]
The integer $B$ is at most $2R_n+1$, so packing after $-R_n-1$
again gives the required envelope.
\qed
\end{proof}

\section{Counting the Intervals}
\label{sec:count}
\begin{theorem}
\label{thm:superpattern}
For $n\ge1$, the family $\I_n$ gives a superpattern for $S_n(213)$
of length at most
\begin{equation}
 n+\up(R_n)\bigl(2+\lceil\log_2n\rceil\bigr)
 =O\!\left(n\log(n+1)\,2^{\sqrt{2\log_2n+4}}\right)
 =n^{1+o(1)}.
 \label{eq:count}
\end{equation}
\end{theorem}
\begin{proof}
Lemma~\ref{lem:forest-layers} represents every $n$-vertex ordered forest
inside $[-R_n,R_n]$. Translate by $W=\up(R_n)$. Every grid spacing
$c_i$ divides $W$, and every midpoint now lies in $[0,2W]$. All the
intervals therefore belong to $\I_n$. Lemma~\ref{lem:forest} turns
this family into a superpattern.

Layer $i$ contains exactly $1+2W/c_i$ intervals. The index $i=1$
contributes $1$ to $\sum_i1/c_i$, and each complete dyadic block
$2^{j-1}<i\le2^j$ contributes $1/2$. Consequently,
\[
 |\I_n|=n+2W\sum_{i=1}^n\frac1{\up(i)}
 \le n+W\bigl(2+\lceil\log_2n\rceil\bigr).
\]
Finally, $W<2R_n\le4n\,2^{\sqrt{2\log_2n+4}}$, which gives
the asserted asymptotic bound.
\qed
\end{proof}

\begin{proof}[Theorem~\ref{thm:main}]
For $n\ge4$, apply Theorem~\ref{thm:superpattern} with $n-3$ in place
of $n$, then apply Lemma~\ref{lem:bcde}. Adding three points does not
change the asymptotic bound. The cases $n\le3$ are immediate.
\qed
\end{proof}

The construction specifies a finite list of intervals, a permutation obtained
by sorting their endpoints, and a fixed coordinate transformation.
The interval endpoints use $O(\log n)$ bits. The final stretching may
produce coordinates with many more bits; a bound on the number of points
does not imply a comparable bound on their total binary encoding length.

Every superpattern for $S_n(213)$ has length $\Omega(n\log n)$, as proved
by Bannister et~al.~\cite{bcde}. Our superpattern bound is therefore within a
subpolynomial factor of the known lower bound. This does not give a
superlinear lower bound for arbitrary universal point sets: the reduction
is only one way. Determining whether $u(n)=O(n)$ remains open.

\appendix
\section{From Superpatterns to Universal Point Sets}
\label{app:reduction}
This appendix gives the construction and proof of
Lemma~\ref{lem:bcde}, due to Bannister, Cheng, Devanny, and
Eppstein~\cite[Section~3]{bcde}. The reduction applies to an arbitrary
$S_{n-3}(213)$-superpattern; the host need not avoid $213$.

Let $\sigma$ be such a superpattern of length $M$, set $Q=M+3$, and form
\begin{equation}
 \tau=(1,Q,\sigma_1+2,\ldots,\sigma_M+2,2).
 \label{eq:augment}
\end{equation}
The universal point set is
\begin{equation}
 U(\sigma)=\{(i,Q^{\tau_i}):1\le i\le Q\}.
 \label{eq:stretch}
\end{equation}
We prove that these $M+3$ points suffice.

\subsection{The canonical permutation}
It is enough to draw a maximal plane graph $G$ with outer vertices $a,b,z$.
Any simple planar graph on at least three vertices can be extended to such
a graph without adding vertices; the additional edges can be deleted after
drawing it. Choose a canonical order
$v_1=a,v_2=b,\ldots,v_n=z$, whose existence follows from
de~Fraysseix, Pach, and Pollack~\cite{dfpp}.
For $t\ge3$, the graph induced by the first $t$ vertices has an outer cycle
containing $ab$. Let $P_t$ be the other boundary path from $a$ to $b$.
Each new vertex attaches to a consecutive block of at least two vertices
of $P_{t-1}$; the internal vertices of the block leave the outer boundary.
For the first insertion we use $P_2=ab$.

Construct an ordered spanning tree $T$ as follows. Initially $b$ is a child
of $a$. When $v_t$ is inserted, make it a child of its first earlier neighbor
along $P_{t-1}$, placing it first in that parent's list of children.
Let $p(v)$ be its preorder rank. Let $r(v)$ be its rank in the reverse of
the postorder using the same child order. Equivalently, $r$ is preorder
with every child list reversed.

The boundary path $P_t$ is increasing in $p$. At an insertion, the new first
child occurs immediately after its parent in the preorder of the tree
constructed so far. It is therefore before the other vertices of the
replaced boundary block. The new boundary remains increasing, and later
insertions do not change the relative preorder of existing vertices.

The order given by $r$ is also canonical~\cite[Lemma~15]{bcde}.
One way to see this is to orient each graph edge from its earlier to its
later endpoint in the original canonical order. A tree edge points from
an ancestor to a descendant, and hence increases in $r$.
For a nontree edge $uv$, with $v$ inserted later, one has $p(v)<p(u)$.
The vertices are incomparable in $T$, so reversing the child lists gives
$r(u)<r(v)$. Thus $r$ is a topological order of this orientation.
Consecutive incomparable vertices can be swapped in a canonical order.
After the first is inserted, the neighbor block of the second cannot contain
or pass through it. The two insertion blocks therefore meet at most at an
endpoint, and the insertions commute. Their first neighbors are different,
so the ordered tree is unchanged. Such swaps connect the two topological
orders and show that $r$ is canonical with the same boundary preorder
property.

For distinct vertices, the two traversals satisfy
\begin{equation}
 u\text{ is an ancestor of }v
 \quad\Longleftrightarrow\quad p(u)<p(v)\text{ and }r(u)<r(v).
 \label{eq:ancestor}
\end{equation}
Define a permutation $\pi$ by $\pi_{p(v)}=r(v)$.
It avoids $213$: such an occurrence would give two incomparable ancestors
of its last vertex, contradicting~\eqref{eq:ancestor}.
Moreover,
\[
 (p(a),r(a))=(1,1),\quad (p(z),r(z))=(2,n),\quad
 (p(b),r(b))=(n,2).
\]
Indeed, $b$ is the last child of $a$ and never acquires a child, whereas
$z$ is a leaf inserted as the first child of $a$.
The permutation $\pi$ therefore begins with $1,n$ and ends with $2$.
Removing these entries gives a member of $S_{n-3}(213)$, after reducing
the remaining values to ranks. Consequently $\pi$ occurs in $\tau$,
using the added entries in~\eqref{eq:augment} for $a,z,b$.

We will use one more consequence of the boundary property. Number the
vertices now by their $r$-ranks. Suppose $uv$ is an edge and
\begin{equation}
 p(u)<p(w)<p(v),\qquad
 r(w)<t:=\max\{r(u),r(v)\}.
 \label{eq:obstruction}
\end{equation}
Then $w$ is interior to the graph induced by the first $t$ vertices.
If it was already interior, there is nothing to prove. Otherwise it is on
the preceding boundary. Let $x$ be the later endpoint of $uv$, let $y$ be
the other endpoint, and let $L,R$ be the retained ends of the block replaced
when $x$ is inserted. The old and new boundary orders give
\[
 p(L)<p(x)<p(R),\qquad p(L)\le p(y)\le p(R).
\]
Since $p(w)$ lies strictly between $p(x)$ and $p(y)$, it lies strictly
between $p(L)$ and $p(R)$. Thus $w$ leaves the boundary. In particular,
it cannot be adjacent to a vertex inserted after time $t$.

\subsection{Why the stretching prevents crossings}
Write $X_i=(i,Q^{\tau_i})$. Among any three points, the highest lies
strictly above the line through the other two. For a direct verification,
let $h<j$ and suppose $X_k$ is the highest. The determinant testing whether
it is above the line $X_hX_j$ is
\[
 D=(j-h)Q^{\tau_k}+(k-j)Q^{\tau_h}+(h-k)Q^{\tau_j}.
\]
If $h<k<j$, then
$D\ge(j-h)(Q^{\tau_k}-Q^{\tau_k-1})>0$.
Otherwise only one of the last two terms can be negative, and its coefficient
has absolute value at most $Q-1$. Hence
$D\ge Q^{\tau_k}-(Q-1)Q^{\tau_k-1}>0$.
This also proves that no three points are collinear.

For four distinct points, suppose $X_k$ is the highest and $h<j$.
The preceding orientation rule gives
\begin{equation}
 X_hX_j\text{ crosses }X_iX_k
 \quad\Longleftrightarrow\quad
 h<i<j\text{ and }\tau_i<\max\{\tau_h,\tau_j\}.
 \label{eq:crossing}
\end{equation}
To verify this, the points $X_h$ and $X_j$ lie on opposite sides of
the line $X_iX_k$ exactly when $h<i<j$. Under that condition, $X_k$ lies
above the line $X_hX_j$, while $X_i$ lies below it exactly when it is not
the highest of $X_h,X_i,X_j$. These are precisely the two conditions for
the segments to cross.

Choose the occurrence of $\pi$ in $\tau$ described above. If its entry
indexed by $p(v)$ occurs at position $x(v)$, draw $v$ at
\[
 \bigl(x(v),Q^{\tau_{x(v)}}\bigr).
\]
The horizontal order is the $p$-order, and the height order is the
$r$-order. Suppose two disjoint edges $uv$ and $wz'$ crossed, where $z'$
has the largest $r$-rank among their four endpoints. Exchange $u$ and $v$
if necessary so that $p(u)<p(v)$. Equation~\eqref{eq:crossing} implies
\[
 p(u)<p(w)<p(v),\qquad
 r(w)<\max\{r(u),r(v)\}<r(z').
\]
By~\eqref{eq:obstruction}, $w$ cannot be adjacent to $z'$, a contradiction.
Thus disjoint edges do not cross. General position also excludes a vertex
in the interior of an edge and overlapping incident edges. This proves
Lemma~\ref{lem:bcde}.

\subsection{The size of the coordinates}
Formula~\eqref{eq:stretch} is an explicit description of the point set.
Its largest ordinate is $Q^Q$, which needs $\Theta(Q\log Q)$ bits.
The total length of all expanded ordinates is
$\Theta(Q^2\log Q)$ bits. The permutation and the fixed stretching formula
give a shorter description. No claim of a near-linear bound on expanded
coordinate length or grid area is needed for Theorem~\ref{thm:main}.

\end{document}